\documentclass[conference]{IEEEtran}

\usepackage{cite}
\usepackage{amsmath,amssymb,amsfonts,amsthm}
\usepackage{bm}
\usepackage{algorithm}
\usepackage{algpseudocode}
\usepackage{graphicx}
\usepackage{textcomp}
\usepackage{xcolor}
\usepackage{booktabs}

\newtheorem{theorem}{Theorem}
\newtheorem{corollary}{Corollary}
\newtheorem{definition}{Definition}

\newtheorem{remark}{Remark}

\def\BibTeX{{\rm B\kern-.05em{\sc i\kern-.025em b}\kern-.08em
    T\kern-.1667em\lower.7ex\hbox{E}\kern-.125emX}}
\algrenewcommand\algorithmicrequire{\textbf{Require:}}
\algrenewcommand\algorithmicensure{\textbf{Output:}}
\begin{document}
\pagestyle{empty}
\setlength{\abovedisplayskip}{6pt plus 2pt minus 2pt}
\setlength{\belowdisplayskip}{6pt plus 2pt minus 2pt}
\setlength{\abovedisplayshortskip}{4pt plus 2pt minus 2pt}
\setlength{\belowdisplayshortskip}{4pt plus 2pt minus 2pt}
\setlength{\jot}{8pt}

\title{From Bit to Block: Capacity Achievement via \\ Code Concatenation}

\author{
\IEEEauthorblockN{Bin Zhang}
\IEEEauthorblockA{\textit{Dept. of Computer Science} \\
\textit{Xiangtan University}\\
Xiangtan 411105, China \\
binzhang.it@gmail.com}
}

\maketitle
\thispagestyle{empty}

\begin{abstract}
		This paper shows that bit-level reliability can be converted into block-level reliability for binary codes over BMS channels through Forney's concatenation scheme.  
		The construction concatenates an inner code of rate $R-o(1)$ and vanishing bit-error probability $\epsilon$ with an outer code of rate $1-o(1)$ equipped with a bounded-distance decoder. 
		We show that, when the outer correction radius exceeds $\epsilon$ and a suitable large-deviation condition holds, the concatenated code has rate $R-o(1)$ and vanishing block-error probability.  
		Hence, a family with vanishing bit-error probability for rates below capacity can be converted into a capacity-achieving concatenated-code family.  
		We further show that, when binary BCH codes are used as outer codes, inner codes with bit-error probability
		$\epsilon=o(1/\log\log n)$ can be converted into a capacity-achieving concatenated family.
	\end{abstract}
	
	
	\section{Introduction}
	
	Shannon introduced channel capacity as the largest rate at which messages can be
	transmitted over a noisy channel with vanishing block-error probability
	\cite{shannon1948}. Since then, a central goal of coding theory has been to find
	explicit code families that achieve this limit. Polar codes, introduced by
	Arikan, were the first explicit codes proved to achieve the capacity of binary
	memoryless symmetric (BMS) channels \cite{arikan2009}. On the binary erasure
	channel (BEC), Kudekar, Kumar, Mondelli, Pfister, Sasoglu, and Urbanke proved
	that doubly transitive linear codes achieve capacity under bit-MAP decoding and RM codes achieve capacity under block-MAP decoding~\cite{kudekar2017reed}. For general BMS channels, the capacity-achieving
	behavior of RM codes was established through the bit-error result of Reeves
	and Pfister~\cite{reeves2024reed} and the block-error result of Abbe and
	Sandon~\cite{abbe2020proof}. Related highly symmetric algebraic code families
	have also been studied for erasure channels: Natarajan and Krishnan showed
	that certain Abelian codes achieve BEC capacity \cite{natarajan2022abelian}
	and later introduced Berman codes, a generalization of RM codes that achieves
	BEC capacity \cite{natarajan2023berman}. More recently, Abbe, Li, and Sly
	introduced Tensor Reed--Muller codes and showed that they achieve capacity \cite{abbe2026tensor}.
	
	In several capacity-achieving results, bit-level reliability is established
	before block-level reliability.  For example, the BEC result above first proves
	vanishing bit-error probability below capacity under bit-MAP decoding for
	doubly transitive linear codes, and then derives the block-MAP statement for RM
	codes.  For BMS channels, Reeves and Pfister proved that RM codes have
	vanishing bit-error probability below capacity before Abbe and Sandon
	established the corresponding block-error statement.  This raises a natural
	question: can bit-level reliability be converted into block-level
	reliability?  A direct union bound gives such a conversion only when the
	bit-error probability $\epsilon_n$ satisfies $n\epsilon_n\to0$.  For erasure
	channels, Pfister, Sprumont, and Zemor studied this question by relating the
	block-error threshold of a linear code on the BEC to its bit-error threshold
	\cite{pfister2025bitblock}.
	
	This paper shows that, for BMS channels, Forney's concatenation can
	convert bit-level reliability into block-level reliability.  The construction
	combines an inner code of rate $R-o(1)$ and uniform bit-error probability
	$\epsilon$ with an outer code of rate $1-o(1)$ and bounded-distance
	correction capability.  
	The resulting concatenated code has rate $R-o(1)$ and can be decoded first by the inner decoder and then by the outer decoder.
	We give a sufficient condition for 	this concatenated code to have vanishing block-error
	probability: the outer correction radius must exceed $\epsilon$, and the outer
	block length must make the binomial large-deviation exponent dominate the union
	bound over columns.  Consequently, any inner family with vanishing bit-error
	probability for rates below capacity, if it satisfies this condition, yields a
	capacity-achieving concatenated family.  Moreover, for BCH outer codes, this condition
	holds whenever the inner bit-error probability satisfies
	$\epsilon=o(1/\log\log n)$, where $n$ is the inner block length; this is much
	slower than the $o(1/n)$ decay required by a direct union bound.

	The contributions of this work are summarized as follows:
	\begin{itemize}
		\item It reveals that Forney's concatenation can convert bit-level reliability of the inner code into block-level reliability of the concatenated code while preserving the asymptotic rate, and gives the corresponding sufficient condition.
		
		\item It shows that BCH outer codes support this conversion 
		when the inner bit-error probability vanishes as
		$o(1/\log\log n)$.
	\end{itemize}
	
	The rest of the paper is organized as follows. Section~\ref{sec:preliminaries}
	fixes notation and reviews the required preliminaries. Section~\ref{sec:main-theorem} presents
	the general bit-to-block concatenation theorem.
	Section~\ref{sec:bch-outer} discusses the reliability threshold supported by
	BCH outer codes. Section~\ref{sec:concluding-remarks}
	concludes the paper.
	
	\section{Preliminaries}
	\label{sec:preliminaries}
	For a prime power $q$, we use $\mathbb F_q$ to denote the finite field of size $q$.  
	For a positive integer $n$, we use $\mathbb F_q^n$ to denote the $n$-dimensional vector space over $\mathbb F_q$.  
	In this paper, we focus on $\mathbb F_2$ and $\mathbb F_2^n$.  
	We write $[n]:=\{1,\ldots,n\}$.  For a vector
	$\bm v\in\mathbb F_2^n$, its Hamming weight is denoted by $|\bm v|$.
	We use the standard asymptotic notation $O(\cdot)$, $o(\cdot)$,
	$\Omega(\cdot)$, $\omega(\cdot)$, and $\Theta(\cdot)$: for positive real sequences
	$\{a_\ell\}_{\ell\ge1}$ and $\{b_\ell\}_{\ell\ge1}$,
	$a_\ell=O(b_\ell)$ means that $a_\ell/b_\ell$ is bounded above by a positive constant,
	$a_\ell=o(b_\ell)$ means that
	$\lim_{\ell\to\infty}a_\ell/b_\ell=0$,
	$a_\ell=\Omega(b_\ell)$ means that $a_\ell/b_\ell$ is bounded below by a
	positive constant, $a_\ell=\omega(b_\ell)$ means that
	\begin{equation*}
		\lim_{\ell\to\infty}\frac{a_\ell}{b_\ell}=\infty,
	\end{equation*}
	and $a_\ell=\Theta(b_\ell)$ means that both
	$a_\ell=O(b_\ell)$ and $a_\ell=\Omega(b_\ell)$ hold.  We write
	$a_\ell\asymp b_\ell$ if $a_\ell=\Theta(b_\ell)$.  We also write
	$a_\ell\ll b_\ell$ if $a_\ell=o(b_\ell)$, and
	$a_\ell\gg b_\ell$ if $a_\ell=\omega(b_\ell)$.
	
	\subsection{Definitions on codes}
	
	\begin{definition}[Binary linear code]
		\label{def:code-rate}
		A binary code of block length $n$ is a subset
		$\mathcal C\subseteq\mathbb F_2^n$.  Its elements are called codewords.
		The code is linear if $\mathcal C$ is a linear subspace of
		$\mathbb F_2^n$.  Its rate is
		\begin{equation*}
			R(\mathcal C)=\frac{1}{n}\log_2|\mathcal C|.
		\end{equation*}
		If $\mathcal C$ is linear with dimension $k$, then
		$R(\mathcal C)=k/n$.
		For a channel with output alphabet $\mathcal Y$, a decoder for
		$\mathcal C$ is a mapping
		$\mathsf D:\mathcal Y^n\to\mathbb F_2^n$ whose output is interpreted as an
		estimate of the transmitted codeword.
	\end{definition}
	
	\begin{definition}[Binary memoryless symmetric channel]
		\label{def:bms}
		A binary memoryless symmetric (BMS) channel $W$ has input alphabet
		$\mathbb F_2$, output alphabet $\mathcal Y$, and transition probability
		$W(y|x)$, where $x\in\mathbb F_2$ and $y\in\mathcal Y$.  The channel is
		memoryless, so for a length-$n$ input vector
		$\bm X=(X_1,\ldots,X_n)\in\mathbb F_2^n$, the output vector
		$\bm Y=(Y_1,\ldots,Y_n)\in\mathcal Y^n$ is governed by the product
		transition law
		\begin{equation*}
			W^n(\bm y|\bm x)
			=
			\prod_{i=1}^n W(y_i|x_i).
		\end{equation*}
		The channel is symmetric means that there is
		a bijection $\pi:\mathcal Y\to\mathcal Y$ such that
		$W(y|0)=W(\pi(y)|1)$ for all $y\in\mathcal Y$.  The capacity of $W$ is
		denoted by $C(W)$ and is achieved by the uniform input distribution.
	\end{definition}
	
	\begin{definition}[Block- and bit-error probabilities]
		\label{def:error-probabilities}
		Let $\mathcal C\subseteq\mathbb F_2^n$ be a binary code and let
		$\mathsf D:\mathcal Y^n\to\mathbb F_2^n$ be a decoder for a BMS channel
		$W$ with output alphabet $\mathcal Y$.  Over $W$, the maximal block-error
		probability is
		\begin{equation*}
			P_{\rm block}(\mathcal C,\mathsf D;W)
			=
			\max_{\bm c\in\mathcal C}
			\Pr\{\mathsf D(\bm Y)\ne \bm c\mid \bm X=\bm c\}.
		\end{equation*}
		The maximal bit-error probability is
		\begin{equation*}
			P_{\rm bit}(\mathcal C,\mathsf D;W)
			=
			\max_{\bm c\in\mathcal C}\max_{j\in[n]}
			\Pr\{(\mathsf D(\bm Y))_j\ne c_j\mid \bm X=\bm c\}.
		\end{equation*}
		Thus, saying that $\mathsf D$ has uniform bit-error probability at most
		$\epsilon$ is equivalent to
		$P_{\rm bit}(\mathcal C,\mathsf D;W)\le\epsilon$.
	\end{definition}

	\begin{definition}[Vanishing error probability at asymptotic rate $R$]
		\label{def:reliable-rate}
		Let $\{(\mathcal C_\ell,\mathsf D_\ell)\}_{\ell\ge1}$ be a sequence of
		binary codes and decoders satisfying
		$\lim_{\ell\to\infty}n_\ell=\infty$.  This
		sequence has asymptotic rate $R$ and vanishing block-error probability on
		a BMS channel $W$ if
		\begin{equation*}
			\lim_{\ell\to\infty} R(\mathcal C_\ell)=R
		\end{equation*}
		and
		\begin{equation*}
			\lim_{\ell\to\infty}
			P_{\rm block}(\mathcal C_\ell,\mathsf D_\ell;W)=0.
		\end{equation*}
		It has vanishing bit-error probability on $W$ if
		\begin{equation*}
			\lim_{\ell\to\infty}
			P_{\rm bit}(\mathcal C_\ell,\mathsf D_\ell;W)=0.
		\end{equation*}
		Vanishing block-error probability implies vanishing bit-error probability.
		Conversely, the union bound gives
		\[
			P_{\rm block}(\mathcal C_\ell,\mathsf D_\ell;W)
			\le
			n_\ell P_{\rm bit}(\mathcal C_\ell,\mathsf D_\ell;W).
		\]
		Thus vanishing bit-error probability implies vanishing block-error probability
		whenever
		\[
			\lim_{\ell\to\infty}
			n_\ell P_{\rm bit}(\mathcal C_\ell,\mathsf D_\ell;W)=0.
		\]
		Otherwise, the converse implication need not
		hold.
		
	\end{definition}
	
	\begin{definition}[Capacity achievement]
		\label{def:capacity-achievement}
		A code construction, or a class of code-decoder sequences, achieves
		Shannon capacity on a BMS channel $W$ if, for every target rate
		$R<C(W)$, it contains a sequence with asymptotic rate $R$ and vanishing
		block-error probability on $W$. 
	\end{definition}
	
	\begin{definition}[Bounded-distance decoder]
		\label{def:bounded-distance-decoder}
		Let $\mathcal C\subseteq\mathbb F_2^n$ be a binary code and let $t$ be a
		nonnegative integer.  A decoder
		$\mathsf D:\mathbb F_2^n\to\mathbb F_2^n$ for $\mathcal C$ is a
		$t$-bounded-distance decoder if,
		\begin{equation*}
			\mathsf D(\bm c+\bm e)=\bm c
			\quad
			\text{for all }\bm c\in\mathcal C\text{ and all }
			\bm e\in\mathbb F_2^n
			\text{ with }|\bm e|\le t.
		\end{equation*}
	\end{definition}
	
	\subsection{Forney's Concatenation}
	\label{subsec:forney-concatenation}
	
	Forney's concatenation was introduced as a constructive route to Shannon's
	coding theorem: the goal is to achieve capacity while keeping the complexity of construction,
	encoding, and decoding polynomial in the final block length~\cite{forney1966}.
	
	Let $\mathcal C_1\subseteq\mathbb F_2^{n_1}$ and
	$\mathcal C_2\subseteq\mathbb F_2^{n_2}$ be binary linear codes.  Their
	concatenated code is the row-column code
	\begin{equation*}
		\operatorname{Con}(\mathcal C_2,\mathcal C_1)
		\subseteq \mathbb F_2^{n_2\times n_1},
	\end{equation*}
	whose codewords are binary matrices whose rows belong to $\mathcal C_1$ and
	whose columns belong to $\mathcal C_2$.  More precisely, for
	$\mathbf X\in\mathbb F_2^{n_2\times n_1}$, write
	\begin{equation*}
		\bm X_{i,:}=(X_{i,1},\ldots,X_{i,n_1})\in\mathbb F_2^{n_1}
	\end{equation*}
	for its $i$-th row and
	\begin{equation*}
		\bm X_{:,j}=(X_{1,j},\ldots,X_{n_2,j})\in\mathbb F_2^{n_2}
	\end{equation*}
	for its $j$-th column.  Then
	\begin{equation*}
		\begin{aligned}
		\operatorname{Con}(\mathcal C_2,\mathcal C_1)
		=\{\mathbf X\in\mathbb F_2^{n_2\times n_1}:&
		\ \bm X_{i,:}\in\mathcal C_1,
		\ i\in[n_2],\\
		&\ \bm X_{:,j}\in\mathcal C_2,
		\ j\in[n_1]\}.
		\end{aligned}
	\end{equation*}
	If $\mathcal C_i$ has dimension $k_i$ and rate $R_i$, for $i=1,2$, then
	$\operatorname{Con}(\mathcal C_2,\mathcal C_1)$ has block length
	$N=n_1n_2$, dimension $k_1k_2$, and rate $R_1R_2$.

	\subsection{Chernoff Bounds}
	\label{subsec:chernoff-bounds}
	
	For $0<a,b<1$, let
	\begin{equation*}
		D(a\Vert b)
		=
		a\log\frac{a}{b}+(1-a)\log\frac{1-a}{1-b}
	\end{equation*}
	be the binary relative entropy, with natural logarithms.  The Chernoff bound
	gives, for $a>b$ and $S\sim{\rm Bin}(n,b)$,
	\begin{equation}
		\Pr\{S\ge an\}
		\le
		\exp\{-nD(a\Vert b)\}.
		\label{eq:chernoff-kl}
	\end{equation}
	We also use the estimate
	\begin{equation}
		D(a\Vert b)=a\log\frac{a}{b}-O(a),
		\label{eq:kl-small-estimate}
	\end{equation}
	whenever $a=a_\ell$ and $b=b_\ell$ satisfy $0<b_\ell<a_\ell$,
	$\lim_{\ell\to\infty}a_\ell=0$, and
	$\lim_{\ell\to\infty}a_\ell/b_\ell=\infty$.
	
\section{A General Concatenation Construction}
	\label{sec:main-theorem}
	
	This section shows that Forney's concatenation can transfer bit-level
	reliability of a component code to block-level reliability of the concatenated
	code.
	
	Consider a row code with small uniform bit-error probability and a high-rate
	column code with prescribed bounded-distance correction capability, and form
	their concatenated code as in Section~\ref{subsec:forney-concatenation}.  The
	decoder operates in two stages.  First, the row decoder is applied
	independently to all rows; after this step, each fixed column contains
	independent residual errors across its coordinates.  Second, the column decoder
	corrects these residual errors whenever their number is within its correction
	radius.  The theorem below gives a sufficient condition under which this
	two-stage decoder converts the bit-level reliability of the row code into
	block-level reliability of the concatenated code.
	
	\begin{theorem}[Bit-to-block reliability via concatenation]
		\label{thm:concat-bit-to-block}
		Fix a BMS channel $W$ and a rate $R<C(W)$.  Consider two code--decoder
		sequences
		\[
			\{(\mathcal C_{1,\ell},\mathsf D_{1,\ell})\}_{\ell\ge1},
			\qquad
			\{(\mathcal C_{2,\ell},\mathsf D_{2,\ell})\}_{\ell\ge1},
		\]
		where $\mathcal C_{i,\ell}\subseteq\mathbb F_2^{n_{i,\ell}}$ is binary
		linear for $i\in\{1,2\}$, and suppose that
		\begin{equation*}
			\lim_{\ell\to\infty}R(\mathcal C_{1,\ell})=R,
			\qquad
			\lim_{\ell\to\infty}R(\mathcal C_{2,\ell})=1.
		\end{equation*}
		Suppose that $\mathsf D_{2,\ell}$ corrects every error pattern of weight at
		most $t_\ell$, and define
		\begin{equation*}
			\epsilon_\ell
			=
			P_{\rm bit}(\mathcal C_{1,\ell},\mathsf D_{1,\ell};W),
			\qquad
			\delta_\ell
			=
			\frac{t_\ell}{n_{2,\ell}}.
		\end{equation*}
		Assume that, for all sufficiently large $\ell$,
		\begin{equation*}
			0<\epsilon_\ell<\delta_\ell<1,
		\end{equation*}
		and
		\begin{equation}
			\lim_{\ell\to\infty}
			\frac{n_{2,\ell}
			D(\delta_\ell\Vert\epsilon_\ell)}
			{\log n_{1,\ell}}
			=\infty.
			\label{eq:main-ld-condition}
		\end{equation}
	Let
	$\mathcal C_{{\rm con},\ell}
	=\operatorname{Con}(\mathcal C_{2,\ell},\mathcal C_{1,\ell})$
	and
		let $\mathsf D_{{\rm con},\ell}$ be the decoder in
		Algorithm~\ref{alg:row-column}.  Then
		\begin{equation*}
			\lim_{\ell\to\infty}R(\mathcal C_{{\rm con},\ell})=R
		\end{equation*}
		and
		\begin{equation}
			P_{\rm block}(\mathcal C_{{\rm con},\ell},
			\mathsf D_{{\rm con},\ell};W)
			\le
			n_{1,\ell}
			\exp\{-n_{2,\ell}
			D(\delta_\ell\Vert\epsilon_\ell)\},
			\label{eq:concat-error-bound}
		\end{equation}
		where the right-hand side tends to zero.  Consequently, if such sequences
		exist for every $R<C(W)$, then the resulting concatenation construction
		achieves capacity on $W$.
	\end{theorem}

	As in Section~\ref{subsec:forney-concatenation}, subscripts $i,:$ and $:,j$
	denote the $i$-th row and the $j$-th column of a matrix, respectively.
	\begin{algorithm}[t]
		\caption{Two-stage decoder for concatenated code}
		\label{alg:row-column}
		\begin{algorithmic}[1]
			\Require Channel output matrix $\mathbf Y\in\mathcal Y^{n_2\times n_1}$,
			where $\bm y_{i,:}$ denotes its $i$-th row; row decoder
			$\mathsf D_1$ for $\mathcal C_1$; column decoder $\mathsf D_2$ for
			$\mathcal C_2$ correcting $t$ errors.
			\Ensure Decoded concatenated-code matrix $\widetilde{\mathbf X}$.
			\State Initialize the row-decoded matrix
			$\widehat{\mathbf X}\in\mathbb F_2^{n_2\times n_1}$.
			\For{$i=1,\ldots,n_2$}
			\State $\widehat{\bm x}_{i,:}\gets
			\mathsf D_1(\bm y_{i,:})$.
			\EndFor
			\For{$j=1,\ldots,n_1$}
			\State $\widetilde{\bm x}_{:,j}\gets
			\mathsf D_2(\widehat{\bm x}_{:,j})$.
			\EndFor
		\end{algorithmic}
	\end{algorithm}
	
	\begin{proof}[Proof of Theorem~\ref{thm:concat-bit-to-block}]
		Fix $\ell$ and a transmitted codeword
		$\mathbf X\in\mathcal C_{{\rm con},\ell}$.  After the row-decoding step,
		define
		\begin{equation*}
			E_{i,j}=\mathbf 1\{\widehat X_{i,j}\ne X_{i,j}\}
		\end{equation*}
		and the residual column-error count
		\begin{equation*}
			S_j=\sum_{i=1}^{n_{2,\ell}}
			E_{i,j},
			\qquad j\in[n_{1,\ell}].
		\end{equation*}
		For each fixed $j$, the variables
		$E_{1,j},\ldots,E_{n_{2,\ell},j}$ are independent because the channel is
		memoryless across rows and $\mathsf D_{1,\ell}$ acts separately on each row.
		Moreover, the definition of $\epsilon_\ell$ gives
		$\Pr\{E_{i,j}=1\}\le\epsilon_\ell$ for all $i$ and $j$.  Therefore $S_j$
		is stochastically dominated by
		${\rm Bin}(n_{2,\ell},\epsilon_\ell)$.
		
		Since $t_\ell=\delta_\ell n_{2,\ell}$ and
		$\delta_\ell>\epsilon_\ell$, the Chernoff bound in~\eqref{eq:chernoff-kl}
		gives
		\begin{equation*}
			\Pr\left\{S_j>t_\ell\right\}
			\le
			\exp\left\{-n_{2,\ell}
			D(\delta_\ell\Vert\epsilon_\ell)\right\}.
		\end{equation*}
		Let $F_j=\{S_j>t_\ell\}$ be the event that column $j$ exceeds the
		correction radius.  If none of the events
		$F_1,\ldots,F_{n_{1,\ell}}$ occurs, then every column decoder restores the
		transmitted column exactly.  Hence
		\begin{equation*}
			P_{\rm block}(\mathcal C_{{\rm con},\ell},\mathsf D_{{\rm con},\ell};W)
			\le
			\Pr\left\{\bigcup_{j=1}^{n_{1,\ell}}F_j\right\}
			\le
			\sum_{j=1}^{n_{1,\ell}}\Pr\{F_j\},
		\end{equation*}
		which proves~\eqref{eq:concat-error-bound}. Note that independence among the events
		$F_j$ is not required for using the union bound.  By~\eqref{eq:main-ld-condition},
		\begin{equation*}
			\log n_{1,\ell}
			-n_{2,\ell}D(\delta_\ell\Vert\epsilon_\ell)
			\longrightarrow-\infty,
		\end{equation*}
		so the bound in~\eqref{eq:concat-error-bound} tends to zero.

		Finally, the concatenated-code dimension formula gives
		\begin{equation*}
			R(\mathcal C_{{\rm con},\ell})
			=R(\mathcal C_{1,\ell})R(\mathcal C_{2,\ell})
			\longrightarrow R.
		\end{equation*}
		If the assumed component sequences exist for every $R<C(W)$, the
		construction achieves capacity by Definition~\ref{def:capacity-achievement}.
	\end{proof}
	
	\begin{remark}
		The concatenation construction trades a small amount of rate for stronger
		reliability.  At finite block length, the column component introduces
		redundancy and reduces the concatenated-code rate from
		$R(\mathcal C_{1,\ell})$ to
		$R(\mathcal C_{1,\ell})R(\mathcal C_{2,\ell})$.  Since
		$R(\mathcal C_{2,\ell})\to1$, this rate loss is asymptotically negligible.
		The benefit of this extra redundancy is that the column decoder cleans the
		residual errors left by the row decoder, thereby upgrading bit-level
		reliability of the row component to block-level reliability of the
		concatenated code.
	\end{remark}

	\begin{remark}
		The condition $\epsilon_\ell<\delta_\ell$ requires the column
		correction radius to exceed the residual bit-error probability after row
		decoding.  Condition~\eqref{eq:main-ld-condition} further requires the
		resulting large-deviation exponent to dominate the union-bound penalty
		$\log n_{1,\ell}$.  Equivalently,
		\begin{equation*}
			n_{2,\ell}
			=
			\omega\left(
			\frac{\log n_{1,\ell}}
			{D(\delta_\ell\Vert\epsilon_\ell)}
			\right).
		\end{equation*}
		If $\delta_\ell\to0$ and
		$\delta_\ell/\epsilon_\ell\to\infty$, then
		\eqref{eq:kl-small-estimate} yields the approximation
		\begin{equation*}
			n_{2,\ell}
			\gg
			\frac{\log n_{1,\ell}}
			{\delta_\ell
			\log(\delta_\ell/\epsilon_\ell)}.
		\end{equation*}
	\end{remark}

	\begin{remark}
		The theorem can be read as a design rule for choosing the column component
		once the row component is fixed.  For a fixed sequence
		$\{(\mathcal C_{1,\ell},\mathsf D_{1,\ell})\}_{\ell\ge1}$, the quantities
		$n_{1,\ell}$ and
		$\epsilon_\ell=P_{\rm bit}(\mathcal C_{1,\ell},\mathsf D_{1,\ell};W)$
		are known.  First choose a target correction fraction
		$\delta_\ell$ satisfying
		$\epsilon_\ell<\delta_\ell<1$.  Then choose the column length
		$n_{2,\ell}$ so that
		\begin{equation*}
			n_{2,\ell}D(\delta_\ell\Vert\epsilon_\ell)
			\gg
			\log n_{1,\ell}.
		\end{equation*}
		Finally, set $t_\ell=\lceil\delta_\ell n_{2,\ell}\rceil$ and choose a
		column code $\mathcal C_{2,\ell}$ of length $n_{2,\ell}$ and rate
		$1-o(1)$ with a bounded-distance decoder correcting $t_\ell$ errors.
		Under these choices, the concatenated decoder has the block-error bound
		\eqref{eq:concat-error-bound}.
	\end{remark}

	\section{BCH Outer Codes}
	\label{sec:bch-outer}

	This section specializes the outer code in Theorem~\ref{thm:concat-bit-to-block}
	to primitive narrow-sense BCH codes.  We quantify the bit-error decay required
	of the inner code when the outer code is BCH.  The result shows that BCH outer
	codes still convert bit-level reliability into block-level reliability when the
	inner uniform bit-error probability vanishes as slowly as
	$o(1/\log\log n_{1,\ell})$.

	Let $s$ and $t$ be positive
	integers, set $n=2^s-1$, and assume $2t+1\le n$.  Then there exists a binary
	primitive narrow-sense BCH code $\mathcal B\subseteq\mathbb F_2^n$ with
	designed distance $2t+1$, redundancy at most $st$, and a bounded-distance
	decoder correcting every error pattern of Hamming weight at most
	$t$~\cite{hocquenghem1959,bose1960,lin2004error}.  Thus
	\begin{equation}
		R(\mathcal B)\ge 1-\frac{st}{n}.
		\label{eq:bch-rate}
	\end{equation}

	\begin{corollary}[BCH outer codes]
		\label{cor:bch-outer}
		Let $\{(\mathcal C_{1,\ell},\mathsf D_{1,\ell})\}_{\ell\ge1}$ be an
		inner code--decoder sequence with inner block length $n_{1,\ell}$ and
		uniform bit-error probability
		\begin{equation*}
			\epsilon_\ell
			=
			P_{\rm bit}(\mathcal C_{1,\ell},\mathsf D_{1,\ell};W).
		\end{equation*}
		If
		\begin{equation}
			\epsilon_\ell\log\log n_{1,\ell}\longrightarrow 0,
			\label{eq:bch-inner-threshold}
		\end{equation}
		then there exists a sequence of primitive narrow-sense BCH outer codes
		$\{\mathcal B_\ell\}_{\ell\ge1}$ of rate tending to one such that the
		corresponding concatenated codes has vanishing block-error probability.
	\end{corollary}

	\begin{proof}
	First choose an sequence $\alpha_\ell$ satisfying
	\begin{equation}
		\epsilon_\ell
		\left(\log\log n_{1,\ell}
		+\log\frac{1}{\epsilon_\ell}\right)
		\ll \alpha_\ell\ll 1 .
		\label{eq:alpha-choice}
	\end{equation}
	Such a choice is possible because~\eqref{eq:bch-inner-threshold} implies
	$\epsilon_\ell\to0$, and hence
	$\epsilon_\ell\log(1/\epsilon_\ell)\to0$.
	
	Choose an integer $s_\ell$ such that
	$s_\ell\asymp \alpha_\ell/\epsilon_\ell$, set
	\begin{equation*}
		n_{2,\ell}=2^{s_\ell}-1,
		\qquad
		t_\ell=\left\lceil 2\epsilon_\ell n_{2,\ell}\right\rceil ,
	\end{equation*}
	and let $\mathcal B_\ell$ be a primitive narrow-sense BCH code of length
	$n_{2,\ell}$ that corrects $t_\ell$ errors.  Since $\epsilon_\ell\to0$, the
	condition $2t_\ell+1\le n_{2,\ell}$ holds for all sufficiently large $\ell$.
	Let
	\[
		\delta_\ell=\frac{t_\ell}{n_{2,\ell}} .
	\]
	Then $\delta_\ell=2\epsilon_\ell+o(\epsilon_\ell)$, so
	$0<\epsilon_\ell<\delta_\ell<1$ eventually.
	
	It remains to verify the two asymptotic hypotheses in
	Theorem~\ref{thm:concat-bit-to-block}.  First, the BCH rate bound gives
	\begin{equation*}
		1-R(\mathcal B_\ell)
		\le
		\frac{s_\ell t_\ell}{n_{2,\ell}}
		\asymp
		s_\ell\epsilon_\ell
		\asymp
		\alpha_\ell
		\longrightarrow0.
	\end{equation*}
	Thus $R(\mathcal B_\ell)\to1$.
	
	Second, since $\delta_\ell=2\epsilon_\ell+o(\epsilon_\ell)$ and $\epsilon_\ell\to0$,
	\[
		D(\delta_\ell\Vert\epsilon_\ell)=\delta_\ell \log(\frac{\delta_\ell}{\epsilon_\ell}) + o(\delta_\ell) = \Theta(\epsilon_\ell).
	\]
	Also, $n_{2,\ell}=2^{s_\ell}-1$ and
	$s_\ell\asymp\alpha_\ell/\epsilon_\ell$, so
	\begin{equation*}
		n_{2,\ell}D(\delta_\ell\Vert\epsilon_\ell)
		=
		\exp\!\left\{\Theta\!\left(\frac{\alpha_\ell}{\epsilon_\ell}\right)\right\}
		\Theta(\epsilon_\ell).
	\end{equation*}
	By~\eqref{eq:alpha-choice},
	\[
		\frac{\alpha_\ell}{\epsilon_\ell}
		\gg
		\log\log n_{1,\ell}+\log\frac{1}{\epsilon_\ell}.
	\]
	Therefore
	\[
	\begin{aligned}
		n_{2,\ell}D(\delta_\ell\Vert\epsilon_\ell)
		&=
		\exp\!\left\{\Theta\!\left(\frac{\alpha_\ell}{\epsilon_\ell}\right)\right\}
		\Theta(\epsilon_\ell) \\
		&=
		\omega\!\left(
		\exp\!\left\{\log\log n_{1,\ell}
		+\log\frac{1}{\epsilon_\ell}\right\}
		\epsilon_\ell
		\right) \\
		&=
		\omega(\log n_{1,\ell}).
	\end{aligned}
	\]
	which is exactly the large-deviation condition
	\eqref{eq:main-ld-condition}.  Theorem~\ref{thm:concat-bit-to-block} then
	gives vanishing block-error probability for the concatenated codes.
	\end{proof}

	\section{Concluding Remarks}
	\label{sec:concluding-remarks}
	
	This paper reveals that Forney's concatenation can convert bit-level
	reliability of an inner code into block-level reliability of the concatenated
	code on BMS channels while preserving the asymptotic rate.  The main theorem
	gives a sufficient condition in terms of the residual bit-error probability,
	the correction radius of the outer code, and the corresponding binomial
	large-deviation exponent.  The BCH specialization further shows that
	rate-one binary BCH outer codes support this conversion whenever the inner
	bit-error probability satisfies $\epsilon=o(1/\log\log n)$.
	
	The result also gives a methodological viewpoint. If a code family can be
	represented, embedded, or approximated through a reliable inner component and a
	high-rate outer code, then block-level reliability of the larger family may follow
	from the concatenation theorem. This perspective may be useful for tensor-type
	constructions and for algebraic code families whose recursive or multilevel
	structure can be interpreted through concatenation operations.

\end{document}